%% file: AbelianPeriodsDAM_arXiv_v2.tex
\documentclass[preprint,11pt,3p,sort&compress]{elsarticle}

\usepackage{fixltx2e,amsmath,graphicx,amssymb,amsthm}
\usepackage[mathlines,displaymath]{lineno}

\usepackage{ALgo}
\usepackage{amsmath,amssymb}
\usepackage{oubraces}
\usepackage{tabularx}
\usepackage{multirow}
\usepackage{tikz}
\usepackage{here}

\def\PV{\mathcal{P}}
\def\QV{\mathcal{Q}} 
\def\PVW{\mathcal{P}_{w}}
\def\PVWA{\mathcal{P}_{wa}}

\def\pp{\ldotp\ldotp}
\def\nt{\mbox{NewHeap}}

\def\cd3#1{\textbf{\textsf{#1}}}
\def\sa#1{\cd3{#1}}
\def\cut{\textit{Cut}_w}
\def\MW{\mathcal{M}_w} 
\def\GW{\mathcal{G}_w} 
\def\selW{\textit{select}_a}
\def\sel{\textit{select}}
\def\C{\textit{C}_w}
\def\S{\textit{S}_w}
\def\ind{\textit{ind}}

\renewcommand{\epsilon}{\varepsilon}

\newtheorem{theorem}{Theorem}[section]
\newtheorem{proposition}[theorem]{Proposition}
\newtheorem{lemma}[theorem]{Lemma}

\newtheorem{definition}{Definition}
\newtheorem{example}{Example}

\begin{document}

\sloppy

\begin{frontmatter}

\title{\textbf{Algorithms for Computing Abelian Periods of Words}\tnoteref{note1}}
\tnotetext[note1]{Some of the results in this paper were presented at Prague Stringology Conference 2011 \cite{FiLeLePr11}.}

\author[palermo]{Gabriele Fici\corref{cor1}}
\ead{Gabriele.Fici@unipa.it}

\author[rouen]{Thierry Lecroq}
\ead{Thierry.Lecroq@univ-rouen.fr}

\author[rouen]{Arnaud Lefebvre}
\ead{Arnaud.Lefebvre@univ-rouen.fr}

\author[rouen]{\'Elise Prieur-Gaston}
\ead{Elise.Prieur@univ-rouen.fr}

\address[palermo]{Dipartimento di Matematica e Informatica, Universit\`a di Palermo, Italy}

\address[rouen]{Normandie Universit\'e, LITIS EA4108, Universit\'e de Rouen, 76821 Mont-Saint-Aignan Cedex, France}

\cortext[cor1]{Corresponding author.}

\journal{Discrete Applied Mathematics}


\begin{abstract}
Constantinescu and Ilie (Bulletin EATCS 89, 167--170, 2006) introduced
 the notion of an \emph{Abelian period} of a word. A word of length $n$ over an alphabet of size $\sigma$ can have $\Theta(n^{2})$ distinct Abelian periods. The Brute-Force algorithm computes all the Abelian periods of a word in time $O(n^2 \times \sigma)$ using $O(n \times \sigma)$ space. We present an off-line algorithm based on a $\sel$ function having the same worst-case theoretical complexity as the Brute-Force one, but outperforming it in practice. We then present on-line algorithms that also enable to compute all the Abelian periods of all the prefixes of $w$.
\end{abstract}


\end{frontmatter}


\section{Introduction}

An integer $p>0$ is a (classical) period of a word $w$ of length $n$
 if $w[i]=w[i+p]$ for every $1\leqslant i \leqslant n-p$.
Classical periods have been extensively studied in Combinatorics on
 Words~\cite{Lothaire2} due to their direct applications in data compression
 and pattern matching.

The Parikh vector of a word $w$ enumerates the
 cardinality of each letter of the alphabet in $w$. 
 For example, given the alphabet $\Sigma=\{\sa{a,b,c}\}$, the Parikh vector of the
 word $w=\sa{aaba}$ is $(3,1,0)$.
The reader can refer to~\cite{BCFL2011} for a list of applications of Parikh
 vectors.

An integer $p$ is an \emph{Abelian period} of a word $w$ over a finite alphabet $\Sigma=\{a_{1},a_{2},\ldots , a_{\sigma}\}$ if $w$
 can be written as $w=u_0u_1 \cdots u_{k-1}u_k$ where for $0<i<k$ all
 the $u_i$'s have the same Parikh vector $\PV$ such that $\sum_{i=1}^{\sigma}\PV[i]=p$
  and the Parikh vectors of $u_0$ and $u_k$ are contained
 in $\PV$~\cite{CI2006}. For example, the word $w=\sa{ababbbabb}$ can be written as $w=u_{0}u_{1}u_{2}u_{3}$, with $u_{0}=\sa{a}$, $u_{1}=\sa{bab}$, $u_{2}=\sa{bba}$ and $u_{3}=\sa{bb}$, and $3$ is an Abelian period of $w$.
 
This definition of Abelian period matches the one of \emph{weak repetition} (also called
 \emph{Abelian power}) when $u_0$ and $u_k$ are the empty word and
 $k>2$~\cite{Cummings_weakrepetitions}.

In recent years, several efficient algorithms have been designed for an Abelian version of the classical pattern matching problem, called the \emph{Jumbled Pattern Matching} problem~\cite{CFL2009,BuCiFiLi10a,BCFL2011,MR10,BuCiFiLi12a,MR12,BaFiKrLi13,GG13}, defined as the problem of finding the occurrences of a substring in a text up to a permutation of the letters in the substring, i.e., the occurrences of any substring of the text having the same Parikh vector as the pattern. 
However, apart from the greedy off-line algorithm given in~\cite{Cummings_weakrepetitions}, no efficient algorithms are known for computing all the Abelian periods of a given word\footnote{During the publication process of the present article, some papers dealing with the computation of the Abelian periods of a word have been published \cite{FiLeLePGSm12,KRR13}.}.

In this article, we present several off-line and on-line algorithms
 for computing all the Abelian periods of a given word.
In Section~\ref{sec-def} we give some basic definitions and fix the notation.
Section~\ref{sec-off} presents off-line algorithms, while Section~\ref{sec-on}
 presents on-line algorithms.
In Section~\ref{sec-exp} we give some experimental results on execution times.
Finally, Section~\ref{sec-conc} contains conclusions and perspectives. 
 

\section{\label{sec-def}Definitions and notation}

Let $\Sigma=\{a_{1},a_{2},\ldots ,a_{\sigma}\}$ be a finite ordered alphabet of
 cardinality $\sigma$ and $\Sigma^*$ the set of words over $\Sigma$. 
We set $\ind(a_i)=i$ for $1\leqslant i \leqslant \sigma$.
We denote by $|w|$ the length of $w$. We  write $w[i]$ the $i$-th symbol of $w$
 and $w[i\pp j]$ the factor of $w$
 from the $i$-th symbol to the $j$-th symbol included,
 with $1\leqslant i \leqslant j\leqslant |w|$.
We denote by $|w|_a$ the number of occurrences of the letter $a\in\Sigma$
 in the word $w$. 

The \emph{Parikh vector} of a word $w$, denoted by $\PVW$, counts the
 occurrences of each letter of $\Sigma$ in $w$, i.e.,  
 $\PVW=(|w|_{a_{1}},\ldots,|w|_{a_{\sigma}})$.
Notice that two words have the same Parikh vector if and only if
 one is obtained from the other by permuting letters (in other words, one is an anagram of the other).
We denote by $\PVW(i,m)$ the Parikh vector of the factor of length $m$ beginning at
 position $i$ in the word $w$.

Given the Parikh vector $\PVW$ of a word $w$, we denote by $\PVW [i]$ its
 $i$-th component 
 and by $|\PVW|$ its norm, that is the sum of its components.
Thus, for $w\in\Sigma^*$ and $1\leqslant i\leqslant\sigma$, we have
 $\PVW [i]=|w|_{a_i}$ and $|\PVW|=\sum_{i=1}^{\sigma}\PVW[i]=|w|$.
Finally, given two Parikh vectors $\PV,\QV$, we write $\PV\subset \QV$ if
 $\PV[i]\leqslant \QV[i]$
 for every $1\leqslant i\leqslant \sigma$ and $|\PV|<|\QV|$. 
 
\begin{definition}[\hspace{-.4mm}\cite{CI2006}]
\label{def-ap}
A word $w$ has an Abelian period $(h,p)$ if $w=u_0u_1 \cdots u_{k-1}u_k$ such that:

\begin{itemize}
 \item $\PV_{u_{0}}\subset \PV_{u_{1}}=\cdots =\PV_{u_{k-1}}\supset \PV_{u_{k}}$,
 \item $|\PV_{u_{0}}|=h$, $|\PV_{u_{1}}|=p$.
\end{itemize}

\end{definition}

We call $u_0$ and $u_k$ resp.\ the \emph{head} and the
 \emph{tail}  of the Abelian period.
 Notice that the length $t=|u_k|$ of the tail is uniquely determined
 by $h$, $p$ and $|w|$, namely $t=(|w|-h) \bmod p$. 

The following lemma gives an upper bound on the number of Abelian periods
 of a word.
 
\begin{lemma}
\label{lemma-max}
A word of length $n$ over an alphabet $\Sigma$ of cardinality $\sigma$ can have $\Theta(n^2)$ different Abelian periods.
\end{lemma}

\begin{proof}
The word $w=(a_{1}a_{2}\cdots a_{\sigma})^{n/\sigma}$ has Abelian period
 $(h,p)$ for any $p\equiv 0 \bmod \sigma$ and every $h$ such that 
 $0 \leqslant h \leqslant \min(p-1,n-p)$. Therefore, $w$ has $\Theta(n^2)$ different Abelian periods.
\end{proof}

A natural order can be defined on the Abelian periods of a word.

\begin{definition}
Two distinct Abelian periods $(h,p)$ and $(h',p')$ of a word $w$ are
 ordered as follows:
 $(h,p) < (h',p') \mbox{ if }
  p<p' 
\mbox{ or } 
 (p=p' \mbox{ and } h<h')$.
\end{definition}

We are interested in computing all the Abelian periods of a word. 
However, the algorithms we present in this paper can be easily adapted to give 
 the smallest Abelian period only.


\section{\label{sec-off}Off-line algorithms}


\subsection{Brute-Force algorithm}

In \figurename~\ref{algo-off-line}, we present a Brute-Force algorithm computing
 all the Abelian periods of an input word $w$ of length $n$.
For each possible head of length $h$ from $1$ to $\lfloor (n-1)/2\rfloor$ the
 algorithm tests all the possible values of $p$ such that $p>h$ and
 $h+p\leqslant n$.
It is a reformulation of the algorithm given
 in~\cite{Cummings_weakrepetitions}.

\begin{figure}[h]
\begin{center}
 \input{algo_offline_reduced.tex}
\caption{
\label{algo-off-line}
Brute-Force algorithm for computing all the Abelian periods of a word $w$ of
 length $n$.
}
\end{center}
\end{figure}

\begin{example}
\label{example1}
For $w=\sa{abaababa}$ the algorithm outputs
 $(1,2)$,
 $(0,3)$, $(2,3)$,
 $(1,4)$, $(2,4)$, $(3,4)$,
 $(0,5)$, $(1,5)$, $(2,5)$, $(3,5)$,
 $(0,6)$, $(1,6)$, $(2,6)$,
 $(0,7)$, $(1,7)$ and
 $(0,8)$.
Among these periods, $(1,2)$ is the smallest. 
\end{example}

\begin{theorem}\label{theor:bf}
The algorithm \textsc{AbelianPeriod-BruteForce} computes
 all the Abelian periods of a given word of length $n$ in time
 $O(n^2 \times \sigma)$ with $O(n\times\sigma)$ space.
\end{theorem}

\begin{proof}
The correctness of the algorithm comes directly
 from Definition~\ref{def-ap}. In a preprocessing phase, all the prefixes of the word are  computed and stored in a table. This takes time $O(n)$ and space $O(n\times \sigma)$. In this way, the computation of the Parikh vector of a factor of the word can be done by computing the difference between two Parikh vectors in the table.
 Since the algorithm performs $\sum_{h=0}^{\lfloor (n-1)/2 \rfloor} \sum_{p=h+1}^{n-h} n/p=O(\sum_{h=1}^n \sum_{p=h}^n n/p) = O(n^{2})$ many comparisons between two Parikh vectors, and since each comparison takes $O(\sigma)$ time, the overall time and space complexity are as claimed (output periods are not stored).
\end{proof}


\subsection{Select-based algorithm}

Let us introduce the $\sel$ function \cite{NavarroMakinen} defined as follows.
\begin{definition}
Let $w$ be a word of length $n$ over alphabet $\Sigma$, then
 $\forall\, a\in\Sigma$:
\begin{itemize}
\item $\selW(w,0)=0$;
\item  $\forall\, 1\leqslant i\leqslant\vert w\vert_a$, $\selW(w,i)=j$ if and only if 
 $j$ is the position of the $i$-th occurrence of letter $a$ in $w$;
 \item $\forall\, i>\vert w\vert_a$, $\selW(w,i)$ is undefined.
\end{itemize}
\end{definition}

In order to compute the $\sel$ function of a word $w$, we consider an array $\S$
 of size $|w|$ storing the (ordered) positions of the occurrences of the letter $a_1$ in $w$,
 then the positions of the occurrences of the letter $a_2$ and so on, up to the positions of the occurrences of the letter  $a_\sigma$. 
In addition to $\S$, we also consider an array $\C$ of $\sigma+1$ elements defined by: $\C[1]=1$, 
 $\C[i]=\sum_{j=1}^{i-1} |w|_{a_j} +1$
 for $1< i \leqslant \sigma$ 
 and  $\C[\sigma +1]=|w|+1$. In fact, $\C[i]-1$ is the number of occurrences of letters strictly smaller than $a_i$ in $w$.
Array $\C$ serves as an index to access $\S$.
Hence, for a letter $a\in \Sigma$ and $i>0$, we have:
$$\sel_{a}(w,i) =
 \begin{cases} 
    \S[\C[\ind(a)]+i-1] & \mbox{ if } i\leqslant \C[\ind(a)+1]-\C[\ind(a)], \\
    \mbox{ undefined} & \mbox{ otherwise.} 
 \end{cases}
$$
\begin{example}
\label{example3}
For $w=\sa{abaababa}$, the $\sel$ function uses the following three arrays:

\begin{center}
\begin{tabular}{lcccccccc}
&$1$&$2$&$3$&$4$&$5$&$6$&$7$&$8$\\
\cline{2-9}
$w$&\multicolumn{1}{|c}{\tt a}&\multicolumn{1}{|c}{\tt b}&\multicolumn{1}{|c}{\tt a}&\multicolumn{1}{|c}{\tt a}&\multicolumn{1}{|c}{\tt b}&\multicolumn{1}{|c}{\tt a}&
\multicolumn{1}{|c}{\tt b}&\multicolumn{1}{|c|}{\tt a}\\
\cline{2-9}
\end{tabular}
\hspace{1cm}
\begin{tabular}{lcc}
&$\tt{a}$&$\tt{b}$\\
\cline{2-3}
$\ind~$&\multicolumn{1}{|c}{$1$}&\multicolumn{1}{|c|}{$2$}\\
\cline{2-3}
\end{tabular}
\hspace{1cm}
\begin{tabular}{lccc}
&{\small $1$}&{\small $2$}&{\small $3$}\\
\cline{2-4}
$\C$&\multicolumn{1}{|c}{$1$}&\multicolumn{1}{|c}{$6$}&\multicolumn{1}{|c|}{$9$}\\
\cline{2-4}
\end{tabular}

\vspace{.5cm}

\begin{tabular}{rcccccccc}
&{\small $1$}&{\small $2$}&{\small $3$}&{\small $4$}&{\small $5$}&{\small $6$}&{\small $7$}&{\small $8$}\\
\cline{2-9}
$\S$&\multicolumn{1}{|c}{$1$}&\multicolumn{1}{|c}{$3$}&\multicolumn{1}{|c}{$4$}&\multicolumn{1}{|c}{$6$}&\multicolumn{1}{|c}{$8$}&\multicolumn{1}{|c}{$2$}&
\multicolumn{1}{|c}{$5$}&\multicolumn{1}{|c|}{$7$}\\
\cline{2-9}\\
\end{tabular}

\end{center}
Then, for instance, $\sel_{\sa{b}}(w,2)=\S[\C[\ind({\sa{b}})]+2-1]=\S[7]=5$, meaning that the second $\sa{b}$ in $w$ appears in position $5$.
\end{example}

Algorithm \textsc{ComputeSelect} (see \figurename{~\ref{algo-computeSel}})
 computes the two arrays $\C$ and $\S$ used by
 the \sel\ function.

\begin{proposition}
Algorithm \textsc{ComputeSelect} runs in
 $O(n+\sigma)$ time and space.
\end{proposition}

\begin{proof}
 The time complexity comes from the fact that
 the \textbf{for} loops in lines
 \ref{algo-computeSel-line1}--\ref{algo-computeSel-line2} and
 \ref{algo-computeSel-line3}--\ref{algo-computeSel-line4}
 are executed $O(\sigma)$ times,
 the \textbf{for}  loop in lines
 \ref{algo-computeSel-line5}--\ref{algo-computeSel-line6}
 is executed $n$ times, and all the other instructions take constant time.
\end{proof}

Once the arrays $C_{w}$ and $S_{w}$ have been computed, each call to the $\sel$ function is
 answered in constant time.

\begin{figure}[h]
\begin{center}
\begin{algo}{ComputeSelect}{w,n}
\SET{\C[1]}{1}
\label{algo-computeSel-line1}
\DOFORI{i}{2}{\sigma+1}
\label{algo-computeSel-line2}
	\SET{\C[i]}{\C[i-1]+\PVW[i-1]}
\OD
\label{algo-computeSel-line3}
\DOFORI{i}{1}{\sigma}
\label{algo-computeSel-line4}
	\SET{P[i]}{0}
\OD
\label{algo-computeSel-line5}
\DOFORI{i}{1}{n}
	\SET{\S[\C[\ind(w[i])]+P[\ind(w[i])]]}{i}
\label{algo-computeSel-line6}
	\SET{P[\ind(w[i])]}{P[\ind(w[i])]+1}
\OD
\RETURN{(\C,\S)}
\end{algo}

\caption{\label{algo-computeSel}Algorithm computing the arrays $\C$ and $\S$.}
\end{center}
\end{figure} 
 
The Brute-Force algorithm tests all possible pairs $(h,p)$, but it is clear
 that, for a given value of $h$, some pairs $(h,p)$ cannot be  Abelian periods of $w$.
For example, let $w=\sa{abaaaaabaa}$ and $h=2$.
Since $\PVW(1,h)$ has to be contained in $\PVW(h+1,p)$, 
 the pairs $(2,3)$, $(2,4)$ and $(2,5)$ cannot be Abelian periods of $w$.
 Indeed, the minimal value of $p$ such that $(2,p)$ can be an Abelian period of $w$ is 
 $6$, in order to contain the second $\sa{b}$ of $w$.
 
This remark leads us to introduce two arrays, $\MW$ and $\GW$, which allow one to skip, for each value $h$ of the head, a number of values of $p$ that are not compatible with $h$.

The array $\MW$ is defined as follows:

\begin{definition}
\label{def-M}
Let $w$ be a word of length $n$ over the alphabet $\Sigma$.
Then $\forall\, 0\leqslant h\leqslant\lfloor(n-1)/2\rfloor$,
$\MW[h]$ is defined by
 $$
\MW[h]=\\
 \begin{cases} \min\{p \mid \PVW(1,h) \subset \PVW(h+1,p) \}& \mbox{ if } \forall\, a \in \Sigma,\, 2\times\vert w[1\pp h]\vert_a \leqslant \vert w\vert_a,\\
      -1 &\mbox{ otherwise.}
 \end{cases}
$$
\end{definition}

In other words, if
$\sel_a(w,2\times\vert w[1\pp h]\vert_a)$ is
 defined for all the letters $a \in \Sigma$, then
$$\MW[h]=\max \{h+1, \max\{\selW(w,2\times\vert w[1\pp h]\vert_a)\mid a\in\Sigma\} -h  \};$$
otherwise, $\MW[h]=-1$.

The algorithm \CALL{ComputeM}{w,n,\C,\S} (see \figurename{~\ref{algo-computeM}})
 builds the array $\MW[h]$ processing the positions of $w$ from left to right.
 
 \begin{figure}[h]
\begin{center}
\begin{algo}{ComputeM}{w,n,\C,\S}
\SET{\MW[0]}{0}
\label{algoM-line1}
\DOFOR{a\in\Sigma}
\label{algoM-line2}
	\SET{H[a]}{0}
\OD
\label{algoM-line3}
\DOFORI{h}{1}{\lfloor\frac{n-1}{2}\rfloor}
	\SET{H[w[h]]}{H[w[h]]+1}
	\SET{s}{\sel_{w[h]}(w,2\times H[w[h]])}
	\IF{s \mbox{ is defined }}
		\SET{\MW[h]}{\max\{\MW[h-1]-1,s-h\}}
\label{algoM-line4}
	\ELSE
		\SET{\MW[h]}{-1}
	\FI
\OD
\label{algoM-line5}
\DOFORI{h}{1}{\lfloor\frac{n-1}{2}\rfloor}
	\IF{\MW[h]=h}
\label{algoM-line6}
		\SET{\MW[h]}{h+1}
	\FI
\OD
\RETURN{\MW}
\end{algo}

\caption{\label{algo-computeM}Algorithm computing the array $\MW$.}
\end{center}
\end{figure}

\begin{proposition}
Algorithm \CALL{ComputeM}{w,n,\C,\S} computes the array $\MW$ in time and space $O(n+\sigma)$.
\end{proposition}

\begin{proof}
The correctness of the algorithm comes directly from Definition~\ref{def-M}.
The time complexity comes from the fact that the \textbf{for} loop in lines
 \ref{algoM-line1}--\ref{algoM-line2} is executed $\sigma$ times,
 the  \textbf{for} loops in lines
 \ref{algoM-line3}--\ref{algoM-line4} and \ref{algoM-line5}--\ref{algoM-line6}
 are executed $O(n)$ times, and all the other instructions take constant time.
\end{proof}

\begin{proposition}
\label{prop-minP}
Let $w$ be a word of length $n$ over the alphabet $\Sigma$, and $h$ such that 
 $0\leqslant h\leqslant \lfloor (n-1)/2\rfloor$.
If $\MW[h]=-1$, then $\forall\, h' \geqslant h$ one has $\MW[h']=-1$, and $h'$
 cannot be the length of the head of an Abelian period of $w$.
\end{proposition}

\begin{proof}
If $\MW[h]=-1$, then by definition there exists a letter $a \in \Sigma$ such that
 $2\times\vert w[1\pp h]\vert_a > \vert w\vert_a$.
Therefore, one cannot find a value $p$ such that
 $\vert w[1\pp h]\vert_a \leqslant \vert w[(h+1)\pp (h+p)]\vert_a$.
It is clear that this is also true for any value $h' >h$.
 \end{proof}

The array $\GW$ is defined as follows:

\begin{definition}
\label{def-G}
Let $w$ be a word of length $n$ over the alphabet $\Sigma$.
Then, for every $h$ such that $0\leqslant h\leqslant\lfloor(n-1)/2\rfloor$, 
 $\GW[h]$ is defined by
\[
 \GW[h] = \max\{\selW(w,i+1)-\selW(w,i)\mid a\in\Sigma,\quad h<\selW(w,i)<\selW(w,i+1)\leqslant n\}.
\]

\end{definition}

In fact, $\GW[h]$ is the maximal value $j'-j$ such that $h< j < j'$ and
 $w[j]=w[j']$, for some $j$ and $j'$.

The array $\GW$ can be computed by the algorithm \CALL{ComputeG}{w,n}
 (see \figurename~\ref{algo-computeG}) processing the
 positions of $w$ from right to left.
 
 \begin{figure}[h]
\begin{center}
\begin{algo}{ComputeG}{w,n}
\SET{\GW[n]}{0}
\label{algoG-line1}
\DOFOR{a\in\Sigma}
\label{algoG-line2}
	\SET{T[a]}{0}
\OD
\label{algoG-line3}
\DOFORI{h}{n}{1}
	\IF{T[w[h]]=0}
		\SET{T[w[h]]}{h}
		\SET{\GW[h-1]}{\GW[h]}
	\ELSE
		\SET{d}{T[w[h]]-h}
		\SET{T[w[h]]}{h}
\label{algoG-line4}
		\SET{\GW[h-1]}{max\{\GW[h],d\}}
	\FI
\OD
\RETURN{\GW}
\end{algo} 
\caption{\label{algo-computeG}Algorithm computing the array $\GW$.}
\end{center}
\end{figure}
 
\begin{proposition}
Algorithm \CALL{ComputeG}{w,n} computes the array $\GW$ in time and space
 $O(n+\sigma)$.
\end{proposition}

\begin{proof}
The correctness of the algorithm comes directly from Definition~\ref{def-G}.
The time complexity comes from the fact that the \textbf{for} loop in lines
 \ref{algoG-line1}--\ref{algoG-line2} is executed $\sigma$ times,
 the \textbf{for} loop in lines
 \ref{algoG-line3}--\ref{algoG-line4}
 is executed $n$ times, and all the other instructions take constant time.
\end{proof}

\begin{proposition}
\label{prop-minimum}
Let $w$ be a word of length $n$ over the alphabet $\Sigma$.
For every $h$ such that $0\leqslant h\leqslant\lfloor(n-1)/2\rfloor$, if $p$ is such that $h < p < \max\{\MW[h],\lfloor(\GW[h]+1)/2\rfloor\}$, then $(h,p)$ is not an
 Abelian period of $w$.
\end{proposition}

\begin{proof}
From the definition of $\MW[h]$, it directly follows that if $ p <\MW[h]$,
 then $(h,p)$ cannot be an Abelian period of $w$.

Given $h$, let $a \in \Sigma$ be such that there exists $1\leqslant i < n$ and
 $\selW(w,i+1) - \selW(w,i) = \GW[h]$.
Let $j= \selW(w,i) $ and $j'=\selW(w,i+1)$.
If  $p <\lfloor(\GW[h]+1)/2\rfloor$, setting $k=\min\{k' \mid h+k'p \geqslant j\}$,
 then $h+(k+1)p< j'$ and
 $|w[h+kp+1 \pp h+(k+1)p]|_a = 0$.
Therefore, $(h,p)$ cannot be an Abelian period of $w$ (see \figurename~{\ref{fig-gap}}).
 \end{proof}

\begin{figure}
\begin{center}
\input{fig_gap.tex}
\caption{\label{fig-gap} If the distance between two consecutive $a$'s in $w$
 is greater than $2p$, then $(h,p)$ cannot be an Abelian period of $w$, for any
 $h<p$.}
 \end{center}
\end{figure}
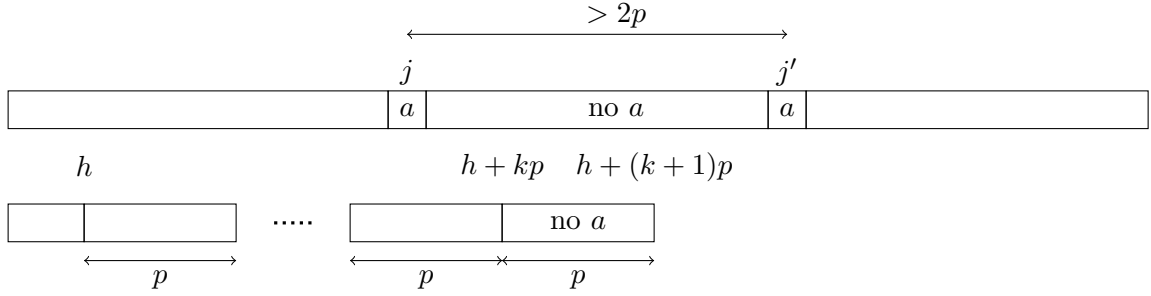

Arrays $\MW$ and $\GW$ give, for every head length $h$, a minimal value for a
 possible $p$ such that $(h,p)$ can be an Abelian period of $w$.
This allows us to skip a number of values for $p$ that cannot give an Abelian
 period. Our next off-line algorithm based on the $\sel$ function will make use of these arrays.

The following lemma shows how to check if $(h,p)$ is an Abelian period
 of $w$ (except for the tail) using the $\sel$ function.

\begin{lemma}
\label{lemma-select}
Let $w$ be a word of length $n$ over the alphabet $\Sigma$.
Let $\mathcal{H}=\PVW(1,h)$ and $\mathcal{P}=\PVW(h+1,p)$.
Let  $i=h+kp$ such that $0<k$, $p\leqslant n-i$, and 
 $(h,p)$ is an Abelian period of $w[1\ldotp\ldotp i]$ (with an empty tail).
Then the following two conditions are equivalent:
\begin{enumerate}
\item $(h,p)$ is an Abelian period of $w[1\ldotp\ldotp i+p]$;
\item for all $a\in\Sigma$
$$
 \selW(w,\mathcal{H}[\ind(a)]+\left(1+\left\lfloor\frac{i}{p}\right\rfloor\right)\times \mathcal{P}[\ind(a)])\leqslant i+p.
$$
\end{enumerate}
\end{lemma} 
 
\begin{proof}
Since $(h,p)$ is an Abelian period of $w[1\ldotp\ldotp i]$ with $i=h+kp$ for
 some $k>0$, then
 $|w[1\ldotp\ldotp i]|_a=\mathcal{H}[\ind(a)]+k\times \mathcal{P}[\ind(a)]$ 
 for each letter $a\in\Sigma$. Notice that since $h<p$ we have $k=\lfloor i/p\rfloor$.\\
($1\Rightarrow 2$).
The fact that $(h,p)$ is an Abelian period of
 $w[1\ldotp\ldotp i+p]$ implies that, for all $a\in\Sigma$,
 $|w[1\ldotp\ldotp i+p]|_a=\mathcal{H}[\ind(a)]+(k+1)\times \mathcal{P}[\ind(a)]$.
Thus, by definition of \sel, we have 
 $\selW(w,\mathcal{H}[\ind(a)]+(1+\lfloor i/p\rfloor)\times \mathcal{P}[\ind(a)])\leqslant i+p$.\\
($2\Rightarrow 1$).
The fact that
 $\selW(w,\mathcal{H}[\ind(a)]+(1+\lfloor i/p\rfloor)\times \mathcal{P}[\ind(a)])\leqslant i+p$
 implies that
 $|w[1\ldotp\ldotp i+p]|_a=\mathcal{H}[\ind(a)]+(k+1)\times \mathcal{P}[\ind(a)]$.
We know that
 $|w[1\ldotp\ldotp i]|_a=\mathcal{H}[\ind(a)]+k\times \mathcal{P}[\ind(a)]$.
By difference, $|w[i+1\ldotp\ldotp i+p]|_a=\mathcal{P}[\ind(a)]$.
Since this is true for all $a\in\Sigma$, we have 
 $\PVW(i+1,p)=\mathcal{P}$, and therefore $(h,p)$ is an Abelian period
 of $w[1\ldotp\ldotp i+p]$.
 \end{proof} 

\figurename~\ref{algo-shift} presents the algorithm \textsc{AbelianPeriod-Shift}
 based on the previous lemma.  
 
\begin{figure}[h]
\begin{center}
\input{shift.tex}
\caption{\label{algo-shift}Algorithm checking whether $(h,p)$ is an Abelian
 period of the prefix of $w$ of 
 length $ n- ((n-h) \bmod p)$.}
 \end{center}
\end{figure} 
 
\begin{proposition}\label{prop:APS}
\label{prop-shift}
Algorithm \CALL{AbelianPeriod-Shift}{h,p,w,n,\C,\S} returns {\sc true} if and only if 
 $(h,p)$ is an Abelian period of the prefix of $w$ of 
 length $ n- ((n-h) \bmod p)$ in time $O(n/p\times\sigma)$ and space $O(\sigma)$.
\end{proposition}

\begin{proof}
The correctness comes directly from Lemma~\ref{lemma-select}.
The {\bf while} loop in line~\ref{boucle-w-i} is executed $n/p$ times and the
 {\bf for} loop in line~\ref{boucle-f-a} is executed $\sigma$ times,
 thus the time complexity is $O(n/p\times\sigma)$. The algorithm only needs to access two Parikh vectors, so the space used is $O(\sigma)$.
\end{proof}

Using Propositions~\ref{prop-minimum} and~\ref{prop-shift},
 algorithm \textsc{AbelianPeriod-Select}, given in Figure~\ref{algo-select}, 
 computes all the Abelian periods of a word $w$ of length $n$.
 
\begin{figure}[h]
\begin{center}
\input{algo_offline_select.tex}
\caption{\label{algo-select}Algorithm computing all the Abelian periods of
 word $w$ of length $n$, based on the $\sel$ function.}
\end{center}
\end{figure}

\begin{theorem}
Algorithm \textsc{AbelianPeriod-Select} computes all the Abelian periods of a
 word of length $n$ in 
 time $O(n^2\times \sigma)$ and space $O(n \times \sigma)$.
\end{theorem}

\begin{proof}
The correctness of the algorithm comes from Propositions~\ref{prop-minimum} and~\ref{prop-shift}. According to Proposition~\ref{prop-minimum}, the value of $p$ computed in
 line~\ref{ligne-max} is the 
 minimal value such that $(h,p)$ can be an Abelian period of the word.
The $\sel$ function is computed through the arrays $\MW$ and $\GW$, which can both be computed in a preprocessing phase in
 $O(n+\sigma)$ time and space.
Always during the preprocessing phase, all the prefixes of the word are  computed and stored in a table. This takes time $O(n)$ and space $O(n\times \sigma)$. In this way, the computation of the Parikh vector of a factor of the word can be done by computing the difference between two Parikh vectors in the table. Therefore, comparing the Parikh vectors of two factors takes $O(\sigma)$ time and space.

The test whether a pair $(h,p)$ is an Abelian period of the word is done by calling the function \textsc{AbelianPeriod-Shift} in line 8 and, if this returns {\sc true}, by verifying the compatibility of the tail in line 10 (output periods are not stored). By Proposition \ref{prop:APS}, each call to the function \textsc{AbelianPeriod-Shift} takes $O(n/p\times \sigma)$ time, while the test on the tail is performed in $O(\sigma)$ time. Thus, the overall time complexity of the algorithm \textsc{AbelianPeriod-Select} is $O(\sum_{h=0}^{\lfloor (n-1)/2 \rfloor} \sum_{p=h+1}^{n-h} (n/p \times \sigma)) =O(\sum_{h=1}^n \sum_{p=h}^n (n/p \times \sigma)) = O(n^{2}\times \sigma)$.

The space needed by the preprocessing phase is $O(n\times \sigma)$, needed for the computation of the table of the Parikh vectors of the prefixes of $w$, whereas the \textbf{while} loop in lines 5--13 only requires $O(\sigma)$ additional space.
 \end{proof}


\section{\label{sec-on}On-line algorithms}

We now propose three on-line algorithms to compute all the Abelian
 periods of a word $w$ using dynamic programming. The idea is to find combinatorial constraints to determine which Abelian periods of the prefix $w[1\pp i-1]$ are still Abelian periods of the prefix $w[1\pp i]$. Moreover, one has to store efficiently the Abelian periods of the prefixes of $w$. The three algorithms we describe below use for this purpose three different data structure: a two-dimensional table, lists and heaps, respectively.

The following proposition states that if $(h,p)$ is not an Abelian period
 of a word $w$, with $h+p\leqslant n=\vert w\vert$, then it cannot be an
 Abelian period of any word having $w$ as a prefix.
 
\begin{proposition}
\label{prop-dp}
Let $w$ be a word of length $n$ and let $h,p$ such that $h+p\leqslant n$.
If $(h,p)$ is not an Abelian period of $w$,
 then $(h,p)$ is not an Abelian period of $wa$ 
 for any letter $a\in\Sigma$.
\end{proposition}

\begin{proof}
If $(h,p)$ is not an Abelian period of $w$, at least
 one of the following three cases holds:
\begin{enumerate}
 \item
 $\PVW(1,h) \not\subset \PVW(h+1,p)$;
 \item there exist two distinct indices $h\leqslant i,i' \leqslant |w|-p+1$
 such
 that $i=kp+h+1$ and $i'=k'p+h+1$ with $k$ and $k'$ two integers
 and $\PVW(i,p) \ne \PVW(i',p)$;
 \item $t = (|w|-h) \bmod p$ and
 $\PVW(|w|-t+1,t) \not\subset \PVW(|w|-p-t+1,p)$.
\end{enumerate}
If case 1 holds, then $\PVWA(1,h) \not\subset \PVWA(h+1,p)$;
if case 2 holds, then $\PVWA(i,p) \ne \PVWA(i',p)$; finally, if case 3 holds, then $\PVWA(|w|-t+1,t+1) \nsubseteq \PVWA(|w|-p-t+1,p)$. In all cases, $(h,p)$ is not an Abelian period of $wa$.
 \end{proof}


\subsection{Two-dimensional array}

The first algorithm (given in \figurename~\ref{algo-on-line-array}) 
  uses a two-dimensional array $T$ and the property stated in Proposition~\ref{prop-dp} to compute all the Abelian periods
 of an input word $w$ on-line.
The algorithm processes the positions of $w$ in increasing order from left to right.
When processing position $i$,
 the value of $T[h,p]$ is set to $j$ if and only if $w[1\pp j]$ is the longest prefix of $w[1\pp i]$ having
 Abelian period $(h,p)$.
Hence, if $j=i-1$, the algorithm checks whether $w[1\pp i]$ has
 Abelian period $(h,p)$, and updates $T[h,p]$
 accordingly.
 
 \begin{figure}[h]
\begin{center}
\begin{algo}[eqtab]{AbelianPeriod-array}{w,n}
\SET{T[0,1]}{1}
\label{algo2-line1}
\DOFORI{i}{2}{n}
  \DOFORI{p}{1}{i-1}
\label{algo2-line2}
    \DOFORI{h}{0}{\mbox{min}\{p-1,i-p-1\}}
	\IF{T[h,p]=i-1}
	  \SET{d}{(i-h)\mod\ell}
	  \IF{d\ne 0}
		\IF{\PVW(i-d+1,d)\subset\PVW(i-d-p+1,p)}
			\SET{T[h,p]}{i}
		\FI
	  \ELSE
		\IF{\PVW(i-p+1,p)=\PVW(i-2\times p+1,p)}
			\SET{T[h,p]}{i}
		\FI
	  \FI
	\FI
    \OD
  \OD
  \DOFORI{h}{0}{\lfloor i/2\rfloor-1}
    \IF{\PVW(1,h)\subset\PVW(h+1,i-h)}
      \SET{T[h,i-h]}{i}
    \ELSE
      \SET{T[h,i-h]}{-1}
    \FI
  \OD
\OD
\RETURN T
\end{algo}
\caption{
\label{algo-on-line-array}
On-line dynamic programming algorithm for computing all the Abelian periods
 of a word $w$ of
 length $n$ using a two-dimensional array.
}
\end{center}
\end{figure}

Since $T[h,p]=i$ if and only if $w[1\pp i]$ is the longest prefix
 of $w$ having Abelian period $(h,p)$, one has that $(h,p)$ is an Abelian period of $w$ if and only if $T[h,p]=n$.

\begin{example}
\label{example2}
For $w=\sa{abaababa}$, the algorithm computes the following array $T$:

\begin{center}
\begin{tabular}{|c|cccccccc|}
\hline
$h\backslash p$ & $1$ & $2$ & $3$ & $4$ & $5$ & $6$ & $7$ & $8$\\
\hline
$0$ & $1$ & $3$ & $8$ & $6$ & $8$ & $8$ & $8$ & $8$\\
$1$ &   & $8$ & $6$ & $8$ & $8$ & $8$ & $8$ &\\
$2$ &   &   & $8$ & $8$ & $8$ & $8$ &   &\\
$3$ &   &   &   & $8$ & $8$ &   &   &\\
\hline
\end{tabular}
\end{center}
Cells in which $T[h,p]=|w|=8$ correspond to pairs $(h,p)$ output by algorithm
 \textsc{AbelianPeriod-BruteForce} of Example~\ref{example1}.
Empty cells on the left part of the array correspond to cases
 in which $h \geqslant p$, while empty cells on the right part
 correspond to cases in which $h+p > |w|$.
\end{example}

\begin{theorem}
The algorithm \textsc{AbelianPeriod-array} computes all the Abelian
 periods of a given word of length $n$ in time
 $\Theta(n^3\times \sigma)$ and space $\Theta(n^2)$.
\end{theorem}

\begin{proof}
The correctness of the algorithm comes from Proposition~\ref{prop-dp}.
The time complexity of the algorithm
 is  due to the three \textbf{for}
 loops of lines~\ref{algo2-line1} to~\ref{algo2-line2}.
The space complexity is due to the array $T$.
\end{proof}


\subsection{Lists}

Also the algorithm \textsc{AbelianPeriod-list}, given in \figurename~\ref{algo-online-list}, processes
 the position of $w$ in increasing order from left to right.
When processing position $i$, it only stores
 the pairs $(h,\ell)$ such that $w[1\pp i-1]$ has Abelian period $\ell$
 with head $h$.
 
 \begin{figure}[h]
\begin{center}
\begin{algo}{AbelianPeriod-list}{w,n}
\SET{L}{\{(0,1)\}}
\label{algo3-line1}
\DOFORI{i}{2}{n}
  \SET{L'}{\emptyset}
\label{algo3-line2}
  \DOFOR{\mbox{all~}(h,\ell)\in L}
    \SET{d}{(i-h)\mod\ell}
    \IF{d\ne 0}
	\IF{\PVW(i-d+1,d)\subset\PVW(i-d-\ell+1,\ell)}
	  \SET{L'}{L'\bigcup\{(h,\ell)\}}
	\FI
    \ELSE
	\IF{\PVW(i-\ell+1,\ell)=\PVW(i-2\times\ell+1,\ell)}
	  \SET{L'}{L'\bigcup\{(h,\ell)\}}
	\FI
    \FI
  \OD
  \SET{L'}{L'\bigcup\{(0,i)\}}
  \DOFORI{h}{1}{\lfloor i/2\rfloor-1}
    \IF{\PVW(1,h)\subset\PVW(h+1,i-h)}
      \SET{L'}{L'\bigcup\{(h,i-h)\}}
    \FI
  \OD
  \SET{L}{L'}
\OD
\end{algo}
\caption{
\label{algo-online-list}
On-line dynamic programming algorithm for computing all the Abelian periods
 of a word $w$ of
 length $n$ using lists.
}
\end{center}
\end{figure}

\begin{theorem}
The algorithm \textsc{AbelianPeriod-list} computes all the Abelian
 periods of a given word of length $n$ in time
 $O(n^3\times \sigma)$ and space $O(n^2)$.
\end{theorem}

\begin{proof}
The correctness of the algorithm comes from Proposition~\ref{prop-dp}.
The space complexity for the list $L$ is given by Lemma~\ref{lemma-max}.
The time complexity of the algorithm
 is  due to the two \textbf{for}
 loops of lines~\ref{algo3-line1} and~\ref{algo3-line2}
 and the maximal number of elements in the list $L$.
\end{proof}


\subsection{Heaps}

The following proposition shows that the set of Abelian periods
 of a prefix of a word can be 
 partitioned into subsets depending of the length of the tail.
In some cases, all the periods of a subset can  be processed
 at once by inspecting only the smallest period of the subset.

\begin{proposition}
\label{prop-heap} 
Let $w$ have $s$ Abelian periods $(h_1,p_1) < (h_2,p_2) < \cdots <(h_s,p_s)$
 such that $(|w|-h_i) \bmod p_i=t>0$ for $1\leqslant i \leqslant s$.
For any letter $a\in\Sigma$, if $(h_1,p_1)$ is an Abelian period of $wa$, then
 $(h_2,p_2),\ldots,(h_s,p_s)$ are also Abelian periods of $wa$.
\end{proposition}

\begin{proof}
Since $(h_1,p_1) < (h_2,p_2) < \cdots <(h_s,p_s)$ are Abelian periods of $w$, we can write
 $w=u_{i,0}u_{i,1}\cdots u_{i,k_i-1}u_{i,k_i}$ with $|u_{i,0}|=h_i$,
 $|u_{i,j}|=p_i$ and $|u_{i,k_i}|=t$ for $1\leqslant i \leqslant s$ and
 $1\leqslant j \leqslant k_i$.
If $(h_1,p_1)$ is an Abelian period of $wa$, then
 $\PV_{u_{1,k_1}a} \subseteq \PV_{u_{1,k_1-1}}$.
Since $|u_{1,k_1}|= |u_{i,k_i}|$ and $|u_{1,k_1-1}|\leqslant |u_{i,k_i-1}|$,
 we have
 that $\PV_{u_{i,k_i}a} \subseteq \PV_{u_{i,k_i-1}}$ for
 $2\leqslant i\leqslant s$.
Hence, $(h_2,p_2),\ldots,(h_s,p_s)$ are Abelian periods of $wa$
 (see~\figurename~\ref{figu-prop-heap}).
 \end{proof}
 
 \begin{figure}[h]
\begin{center}
\input{fig_sync.tex}
\end{center}
\caption{
\label{figu-prop-heap}
$w=u_{i,0}u_{i,1}\cdots u_{i,k_i-1} u_{i,k_i}$, $u_{i,k_i}=z$ for
 $1\leqslant i \leqslant s$.
If $\PV_{za} \subseteq \PV_{u_{1,k_1-1}}$, then
 $\PV_{za} \subseteq \PV_{u_{i,k_i-1}}$ for every $2\leqslant i\leqslant s$.
}
\end{figure}	

The algorithm \textsc{AbelianPeriod-heap}, given in \figurename~\ref{algo-on-line-heap}, uses the property stated in
 Proposition~\ref{prop-heap} for computing all the Abelian
 periods of an input word $w$ by gathering all the ongoing periods $(h,p)$ with the
 same tail length in a heap where the element in the root
  is the smallest period.
 
When processing $w[i]$, the algorithm processes every heap $H$
 for the different tail lengths:
\begin{itemize}
\item
 if the period $(h,p)$ at the root of $H$ is a period of $w[1\pp i]$,
 then by Proposition~\ref{prop-heap} all the elements of $H$
 are Abelian periods of $w[1\pp i]$.
If the tail length becomes equal to $p$, then $(h,p)$ is removed from
 the current heap and is put into a new heap corresponding to the
 empty tail.
\item
if the period $(h,p)$ at the root of $H$ is not a period of $w[1\pp i]$,
 then it is removed from $H$ and the same process is applied until
 a pair $(h',p')$ is an Abelian period of $w[1\pp i]$ or the heap becomes
 empty. This procedure is realized by function \textsc{ExtractUntilOK} in 
line~\ref{algo-heap-extractuntilok}.
\end{itemize}

Finally, all the degenerate cases $(h,p)$ such that $h<p$ and $h+p=i$ have to
 be inserted in the heap corresponding to the empty tail
 (lines~\ref{algo-heap-line-empty-tail1} to~\ref{algo-heap-line-empty-tail2}).
 
The function \CALL{Root}{H} returns the smallest element of the heap $H$,
 the function
 \CALL{Insert}{H,e} inserts element $e$ in the heap $H$, while the function
 \CALL{Remove}{H} removes the smallest element of the heap $H$.

\begin{figure}[h]
\begin{center}
\input{algo_online_heap.tex}
\caption{
\label{algo-on-line-heap}
On-line algorithm for computing all the Abelian periods of a word $w$ of
 length $n$ using heaps.
}
\end{center}
\end{figure}
 
\begin{theorem}
The algorithm \textsc{AbelianPeriod-heap} computes all the Abelian
 periods of a given word of length $n$ in time
 $O(n^3\log n \times\sigma)$ and space $O(n^2)$.
\end{theorem}

\begin{proof}
The space memory depends on the total number of nodes of the heaps.
Since one node corresponds exactly to one Abelian period, the maximum
number of nodes is then bounded by $n^2$.

For the same reason, during each execution of the \textbf{for} loop 
 starting in line 4, the maximum number of nodes removed or inserted by 
\textsc{ExtractUntilOK}, \textsc{Remove} and \textsc{Insert} functions
 is bounded by $n^2$. Each of these functions takes time at most $\log n$.
Comparing two Parikh vectors in line 7 takes time at most $\sigma$.
The time complexity of this loop is then $O(n^2\log n \times\sigma)$.

The \textsc{Insert} function in the \textbf{while} loop starting in line 13 
 is called at most $n$ times. The time complexity of this loop is then
 $O(n\log n)$.

Since these two loops are executed $n$ times (loop \textbf{for} starting in line 2) 
 the time complexity of this algorithm is 
 $O(n^3\log n\times \sigma)$.    
 \end{proof}


\section{\label{sec-exp}Experimental results}

Practical performances of the two off-line algorithms have been compared.
They both have been implemented in C in a homogeneous way using the table of the Parikh vectors of the prefixes of the word, and run on test sets of random words ($3\,000$ words each) of different lengths (from $10$ to $10\,000$) on different alphabet sizes ($2$, $3$, $4$, $8$ and $16$). 
Tests were performed on a MacBook Pro laptop running Mac OS X with a 2.2 GHz processor and 2~GB~RAM.
 
A first remark is that most of the Abelian periods of a word have only one occurrence of the factor of length $p$, i.e., are such that $h+2p \geqslant |w|$. We call these latter \emph{trivial Abelian periods}. To give an idea, the prefix of length $4\,181$ of the Fibonacci word  $F=\sa{abaababaabaab}\cdots$ has $3\,453\,511$ Abelian periods, but only $538\,739$ (i.e., about 15.6\%) are non-trivial. The same proportion holds for longer prefixes of the Fibonacci word. But the Fibonacci word is probably one of the words with the highest proportion of non-trivial Abelian periods. 
Note that the word $\sa{a}^{2\,090}\sa{b}\sa{a}^{2\,090}$ of the same length has $2\,914\,854$ Abelian periods, and all of them are trivial.

If one considers \emph{all} the Abelian periods (that is, both trivial and non-trivial) running times of the two algorithms are very close, and seem to depend on the machine architecture more than on the algorithm itself (results not shown).
If instead one computes non-trivial Abelian periods only, the  $\sel$-based algorithm significantly improves on the Brute-Force one, and the gap increases when the alphabet size increases. In fact, even if the worst-case complexity of the two algorithms depend on $\sigma$, the $\sel$-based algorithm seems to have an average behavior independent from the alphabet size.
In Figure~\ref{fig-time} we show results for alphabet sizes $2$ and $16$. 
These tests also suggest that the $\sel$-based algorithm becomes much faster than the brute-force algorithm when the word length increases. 

\begin{figure}
\begin{center}
\includegraphics[width=13.0cm]{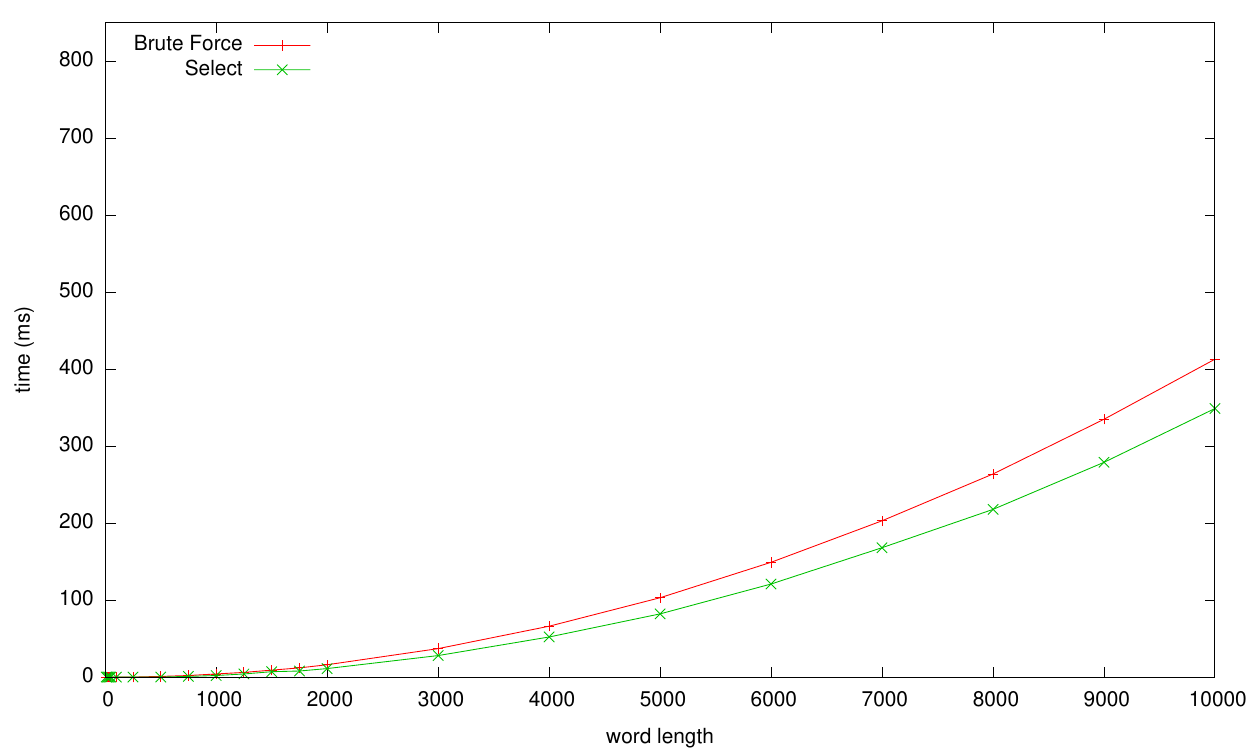}
\includegraphics[width=13.0cm]{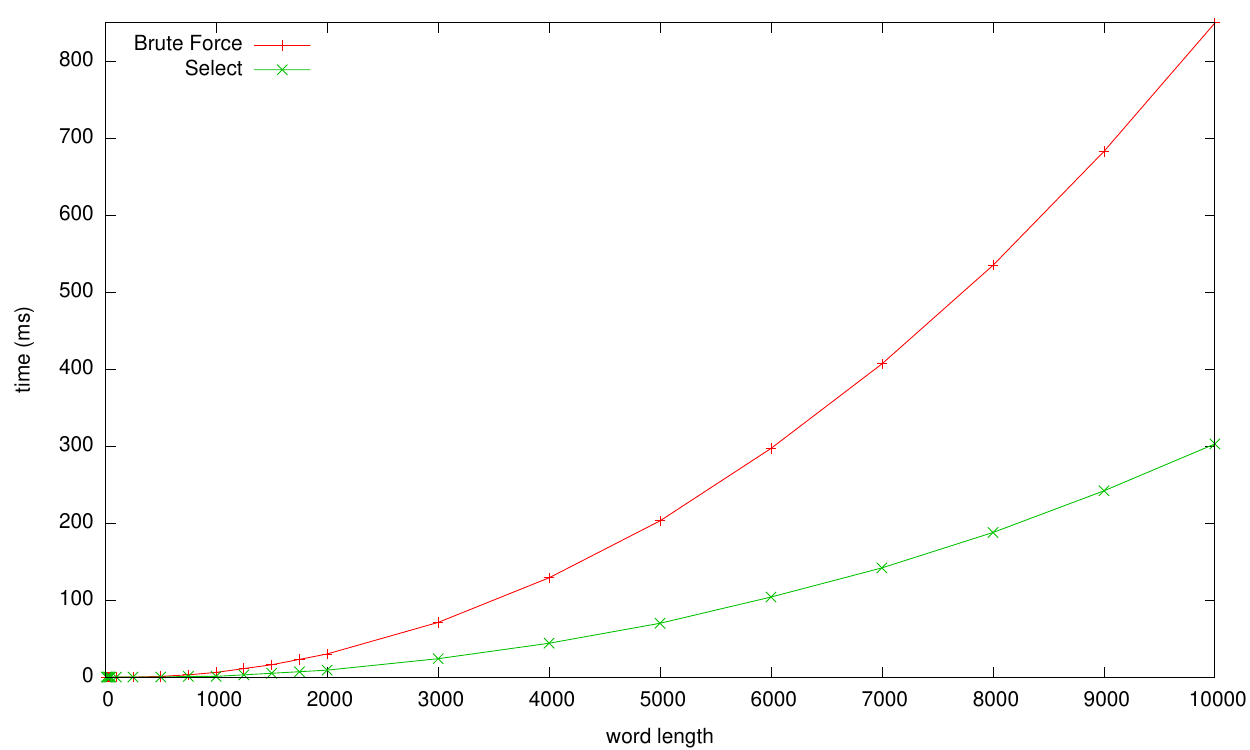}
\end{center}
\caption{\label{fig-time}
Average running times (in ms), over $3\,000$ random words, of the
 Brute-Force and $\sel$-based algorithms
 on alphabet size $2$ (top) and $16$ (bottom), in the case where $h+2p\leqslant |w|$, i.e., for at least two repetitions of the Abelian period.
}
\end{figure}


\section{Conclusion and perspectives}\label{sec-conc}

 This paper is the first attempt to give algorithms for computing all the Abelian
 periods of a word. As shown in Lemma \ref{lemma-max}, the total number of Abelian periods of a word can be quadratic in its length.
 We gave an $O(n^2\times\sigma)$ time off-line algorithm based on the \emph{select} function that in practice appears to be significantly faster than the Brute-Force one, as discussed in the experimental part section. We also presented three on-line algorithms that compute the Abelian periods of all the prefixes of the word. 

However, some Abelian periods exist just as a consequence of the existence of smaller ones. For instance, in the word $w=\sa{abaababa}$ of Example \ref{example1}, the fact that $(1,4),(1,6),(3,4)$ are Abelian periods for $w$ is just a consequence of the fact that $(1,2)$ is. So, let us define the \emph{cutting positions} of an Abelian period $(h,p)$ 
 as follows:
$$\cut(h,p)=\{k=h+jp \mid 1\leqslant k\leqslant |w|\mbox{ and }0\leqslant j\}.$$
We say that an Abelian period $(h,p)$ of $w$ is \emph{non-deducible} if there does not exist
 another Abelian period $(h',p')$
 of $w$ such that $\cut(h,p)\subset\cut(h',p')$.
Anyway, even the number of non-deducible Abelian periods can
 still be quadratic.

It seems quite clear that balanced words (words such that for any letter $a\in \Sigma$ the difference of the number of $a$'s in any two factors of the same length is bounded by $1$) are the words with the maximum number of Abelian periods. In a recent paper, together with Alessio Langiu and Filippo Mignosi \cite{FiLaLeLeMiPG13}, we studied the Abelian repetitions in Sturmian words and gave a formula for computing the smallest Abelian period of the Fibonacci finite words. Preliminary experiments toward this results were done using the algorithms presented in this paper.

On the opposite side, it remains to obtain a bound on the minimal Abelian period given a word
 length and an alphabet size. 
Simple modifications of the presented algorithms would allow one to compute
 the minimal Abelian period of each factor of a word.


\bibliographystyle{plain}

\end{document}

%% file: algo_offline_reduced.tex
\begin{algo}{AbelianPeriod-BruteForce}{w,n}
\DOFORI{h}{0}{\lfloor(n-1)/2\rfloor }
 \SET{p}{h+1}
 \DOWHILE{h+p\leq n }
 \label{algo1-line1}
  \IF{\text{$(h,p)$ is an Abelian period of $w$}}
	\CALL{Output}{h,p}
  \FI
  \SET{p}{p+1}
 \OD
\OD
\end{algo}


%% file: fig_gap.tex
\begin{tikzpicture}

\draw (0,0) rectangle (5,.5);
\draw (5,0) rectangle (5.5,.5);
\draw (5.5,0) rectangle (10,.5);
\draw (10,0) rectangle (10.5,.5);
\draw (10.5,0) rectangle (15,.5);

\node at (5.25,.75) {$j$};
\node at (5.25,.25) {$a$};
\node at (8,.25) {no $a$};
\node at (10.25,.75) {$j'$};
\node at (10.25,.25) {$a$};

\draw (0,-1.5) rectangle (1,-1);
\draw (1,-1.5) rectangle (3,-1);
\draw[very thick,dotted] (3.5,-1.25) -- (4,-1.25);

\draw (4.5,-1.5) rectangle (6.5,-1);
\draw (6.5,-1.5) rectangle (8.5,-1);

\node at (1,-0.5) { $h$};

\node at (6.5,-0.5) {$h+kp$};
\node at (8.5,-0.5) {$h+(k+1)p$};
\node at (7.5,-1.25) {no $a$};

\draw[<->] (1,-1.75) -- (3,-1.75) node[right] {};
\node at (2,-2) {$ p$};

\draw[<->] (5.25,1.25) -- (10.25,1.25) node[right] {};
\node at (8,1.5) {$ >2p$};

\draw[<->] (4.5,-1.75) -- (6.5,-1.75) ;
\node at (5.5,-2) {$p$};

\draw[<->] (6.5,-1.75) -- (8.5,-1.75) ;
\node at (7.5,-2) {$p$};

\end{tikzpicture} 

%% file: shift.tex
\begin{algo}{AbelianPeriod-Shift}{h,p,w,n,\C,\S}
\SET{(\mathcal{H},\mathcal{P})}{(\PVW(1,h),\PVW(h+1,p))}
\SET{i}{h+p}
\label{boucle-w-i}\DOWHILE{i+p\leqslant n}
	\label{boucle-f-a}\DOFOR{a\in\Sigma}
		\SET{s}{\selW(w,\mathcal{H}[\ind(a)]+(1+\lfloor i/p\rfloor)\times \mathcal{P}[\ind(a)])}
		\IF{s\mbox{ is undefined}\OR s > i+p}
			\RETURN{\FALSE}
		\FI
	\OD
	\SET{i}{i+p}
\OD
\RETURN{\TRUE}

\end{algo}

%% file: algo_offline_select.tex
\begin{algo}{AbelianPeriod-Select}{w,n}
\label{ligne-select}
\SET{(\C,\S)}{\CALL{ComputeSelect}{w,n}}
\label{ligne-M}
 \SET{\MW}{\CALL{ComputeM}{w,n,\C,\S}}
\label{ligne-G}
\SET{\GW}{\CALL{ComputeG}{w,n}}
\SET{h}{0}
\DOWHILE{h \leqslant \lfloor(n-1)/2\rfloor \AND \MW[h]\neq -1}
\label{ligne-max}
 \SET{p}{\CALL{$\max$}{\MW[h],\lfloor(\GW[h]+1)/2\rfloor}}
\label{boucle-w}
 \DOWHILE{h+p\leqslant n}
\label{ligne-shift}
	\IF{\CALL{AbelianPeriod-Shift}{h,p,w,n,\C,\S}}
		\SET{t}{(n-h)\bmod p}
		\label{test-tail}
		\IF{\PVW(n-t+1, t) \subset \PVW(h+1,p) }
			\CALL{Output}{h,p}
		\FI
 	\FI
	\SET{p}{p+1}
 \OD
 \SET{h}{h+1}
\OD
\end{algo}

%% file: fig_sync.tex
\begin{tikzpicture}
\draw (0,0) rectangle (5,.5);
\draw (0,1) rectangle (5,1.5);
\draw (0,1.75) rectangle (5,2.25);
\draw (4.5,0) -- (4.5,.5);
\draw (4.5,1) -- (4.5,1.5);
\draw (4.5,1.75) -- (4.5,2.25);
\node at (4.75,.25) {$a$};
\node at (4.75,1.25) {$a$};
\node at (4.75,2) {$a$};
\draw (3.5,0) -- (3.5,.5);
\draw (3.5,1) -- (3.5,1.5);
\draw (3.5,1.75) -- (3.5,2.25);
\node at (4,.25) {$z$};
\node at (4,1.25) {$z$};
\node at (4,2) {$z$};
\draw (1.8,0) -- (1.8,.5);
\node at (2.65,.25) {$u_{s,k_s-1}$};
\draw (2,1) -- (2,1.5);
\node at (2.75,1.25) {$u_{2,k_2-1}$};
\draw (2.2,1.75) -- (2.2,2.25);
\node at (2.85,2) {$u_{1,k_1-1}$};
\node at (2.5,.8){\vdots};
\end{tikzpicture} 

%% file: algo_online_heap.tex
\begin{algo}{AbelianPeriod-heap}{w,n}
\SET{L}{\mbox{list with one heap containing }(0,1)}
\label{algo-heap-line1}
\DOFORI{i}{2}{n}
  \SET{\nt}{\emptyset}
\label{algo-heap-line2}
  \DOFOR{\mbox{all~}H\in L}
    \SET{(h,p)}{\CALL{Root}{H}}
    \SET{t}{p - ((i-h)\bmod p)}
    \IF{\PVW(i-t+1,t)\not\subseteq\PVW(i-t-p+1,p)}
\label{algo-heap-extractuntilok}
       \CALL{ExtractUntilOK}{H}
    \ELSE
      \IF{t=p}
	  \CALL{Remove}{H}
	  \CALL{Insert}{\nt,(h,p)}
      \FI
    \FI
  \OD
  \label{algo-heap-line-empty-tail1}\SET{h}{0}
  \DOWHILE{h<\lfloor (i+1)/2\rfloor \AND \PVW(1,h)\subset\PVW(h+1,i-h)}
      \CALL{Insert}{\nt,(h,i-h)}
      \label{algo-heap-line-empty-tail2}\SET{h}{h+1}
  \OD
  \SET{L}{L\cup\nt}
\OD
\RETURN L
\end{algo}